\documentclass[12pt]{article}
\usepackage[text={14cm, 20cm}]{geometry}
\usepackage{amssymb,amsbsy,float}
\usepackage{graphicx}
\usepackage{amsfonts}
\usepackage{latexsym}
\usepackage{bibtopic}
\usepackage{amsmath}
\usepackage{amsthm}

\newtheorem{thm}{Theorem}[section]
\newtheorem{prop}[thm]{Proposition}
\newtheorem{lem}[thm]{Lemma}

\newtheorem{cl}[thm]{Claim}

\newcommand{\R}{\mathbb{R}}  

\newcommand{\finpreuve}{\ \rule{0.5em}{0.5em} \vspace{0.1 cm}}

\newcommand{\bx}{\mathbf{x}}
\newcommand{\bU}{\mathbf{U}}

\newcommand{\bb}{\mathbf{b}}\newcommand{\by}{\mbox{$\textbf{y}$}}%
\newcommand{\be}{\mathbf{e}}

\newcommand{\bq}{\mathbf{q}}
\newcommand{\bn}{\mathbf{n}}
\newcommand{\bz}{\mathbf{z}}

\newcommand{\eps}{\varepsilon}

\begin{document}

\title{Game dynamics and Nash equilibria}
\author{Yannick Viossat\thanks{
\emph{E-mail:} viossat $\alpha\tau$ ceremade.dauphine.fr. }
\thanks{The author thanks Sylvain Sorin, for patient and painful hours spent trying to decipher a first version of this article. The support of the ANR RISK and of the Fondation du Risque (Chaire Groupama) is gratefully acknowledged. All errors are mine.}}
\date{CEREMADE, Universit{\'e} Paris-Dauphine}

\maketitle








\begin{abstract}
There are games with a unique Nash equilibrium 
but such that, \emph{for almost all initial conditions}, all strategies in the support of this equilibrium are eliminated by the replicator dynamics and the best-reply dynamics. 
\end{abstract}


 \maketitle
\section{Introduction}
\label{sec:intro} Evolutionary game dynamics model the evolution of the mean behavior in populations of agents interacting strategically. 
A most studied topic is the link between the outcome of these dynamics and Nash equilibria. Many positive connections have been found, 
including convergence to the set of Nash equilibria for many dynamics in special classes of games (Sandholm, 2010).  
In general though, solutions of evolutionary game dynamics need not converge to the set of Nash equilibria (Hofbauer and Sigmund, 1998, section 8.6). By contrast with no-regret dynamics (e.g., Hart, 2005 and references therein), replacing Nash equilibria by correlated equilibria and convergence of the solutions by convergence of some time-average hardly helps: for many dynamics, there are examples of games with a unique Nash equilibrium, which is also the unique correlated equilibrium, but such that, for some initial conditions, all strategies in the support of this equilibrium are eliminated (Viossat, 2007, 2008). 

In these examples however, the Nash equilibrium is strict and thus asymptotically stable under reasonable dynamics. This leads to the following question: are there games such that all strategies in the support of Nash equilibria are eliminated \emph{for almost all initial conditions}? This article shows that  the answer is positive, at least for 
the two most studied 
dynamics: the replicator dynamics (REP) and the best-reply dynamics (BR). For BR, we exhibit \emph{an open set} of such games. 


Our examples are relatively high dimensional: $6 \times 6$  games for BR
and $7 \times 7$ for REP. The reason why we need an extra-dimension for the replicator
dynamics 
seems purely technical: our examples for the best-reply dynamics should work as
well for the replicator dynamics, but this is not so easy
to prove, as the replicator dynamics is more difficult to
analyze than the best-reply dynamics.

The reason why our games are relatively high dimensional is
deeper: 
first, by the folk-theorem of evolutionary
game theory (Weibull, 1995, Prop. 4.11), if an interior
trajectory of REP or BR converges to a point, then this point is a
Nash equilibrium. Thus, we need nonconvergent trajectories, and
along which, asymptotically, only strategies that do not belong to
the support of a Nash equilibrium have positive probability. For single population dynamics, this
seems to require at least three strategies not in the support of Nash equilibria. Moreover, the only solution for having a unique strategy in the support of at least one Nash equilibrium is to have a unique, pure Nash equilibrium. But such a
Nash equilibrium would be strict, hence asymptotically stable:
\begin{prop}
\label{prop:strictNash} In a bimatrix game, a unique and pure Nash
equilibrium is strict.
\end{prop}
\begin{proof}
A Nash equilibrium is \emph{quasi-strict} if each player puts positive weight on each of her pure best-replies. In a bimatrix game, if a Nash equilibrium is unique, then it is quasi-strict (Jansen, 1981; Norde, 1999); if it is
unique and pure, it is quasi-strict and pure, hence strict.
\end{proof}

We thus need at least two strategies in the support of Nash equilibria. With the three strategies not in the support of equilibria, 
this makes at least five strategies. Our examples for the best-reply dynamics are $6
\times 6$ games: there might be room for improvement, but
not much.


The remainder of this article is organized as follows: the framework and the notation are introduced below. Section \ref{sec:BR} studies the behavior of the best-reply dynamics in a family of $6 \times 6$ games. Section \ref{sec:REP} studies the replicator dynamics in a specific $7 \times 7$ game. Section \ref{sec:discussion} concludes. Appendix \ref{app:A} shows that the games we study have a unique Nash equilibrium. Appendix \ref{app:BR} studies the behavior of the best-reply dynamics in the $7 \times 7$ game of Section \ref{sec:REP}.\\ 

\noindent \emph{Notation and definitions.} We study
single-population dynamics in two-player, finite symmetric games.
The set of pure strategies is $I=\{1,2,..,N\}$ and the payoff
matrix is $\bU=(u_{ij})_{1 \leq i,j \leq N}$. Thus, $u_{ij}$ is the payoff of an individual playing strategy $i$ 
against an individual playing strategy $j$. Let $S_N$ denote
the simplex of mixed strategies (henceforth, ``the simplex"):
$$S_N=\left\{\bx \in \R_+^{I} : \sum_{i \in I} x_i =1\right\}.$$  
Its vertices $\be_i$, $1 \leq i \leq N$, correspond to the pure strategies of the game. 
Note that vectors and matrices are denoted by bold characters.

Denote by $x_i(t)$ the proportion of the population playing strategy $i$ at time $t$ and by $\bx(t)=(x_1(t),...,x_N(t)) \in S_N$ the population profile (or mean strategy).  
We often omit time arguments and write 
$\bx$ 
for $\bx(t)$. 
We study the evolution of the population profile under the two most studied dynamics: the replicator dynamics and the best-reply dynamics.  

The \emph{replicator dynamics} (Taylor and Jonker, 1978) may be derived by assuming that the per capita growth rate of the total number of individuals playing strategy $i$ is the payoff of the game.\footnote{Or, up to a change of time, a background fitness plus the payoff of the game.} For frequencies of strategies, this leads to: 
\begin{equation}
\tag{REP}
\label{eq:defREP} \dot{x}_{i}=x_{i}\left[(\bU\bx)_{i}-
\bx \cdot \bU\bx \right]
\end{equation}

The right-hand side is Lipschitz in $\bx$, hence there is a unique solution through each initial condition. This solution is \emph{interior} if $x_i(t)>0$ for all $i \in I$ and all $t \in \R$. 
Since the faces of the simplex are invariant under (\ref{eq:defREP}), this boils down to the initial condition being interior; that  is, $x_i(0)>0$ for all $i$ in $I$. 

The \emph{best-reply dynamics} (Gilboa and Matsui, 1991; Matsui, 1992) may be derived by assuming that in each small time interval, a fraction of the population revises its strategy and switches (rationally, but myopically) to a best-reply to the current population profile.  Since this best-reply need not be unique, this does not lead to a differential equation but to the differential \emph{inclusion}:
\begin{equation}%
\tag{BR}
\label{eq:defBR}%
\dot{\bx} \in BR(\bx)-\bx
\end{equation}
where $BR(\bx)=\{\mbox{$\textbf{y}$} \in S_N: \by \cdot \bU\bx
=\max_{\bz \in S_N} \bz \cdot \bU \bx\}$ denotes the set of mixed best-replies to $\bx$. 
A solution of the best-reply dynamics is an absolutely continuous function $\bx: \mathbb{R}_+ \to S_N$ satisfying (\ref{eq:defBR}) 
for almost every $t$. Solutions exist for each initial condition, but need not be unique. 
%

\emph{Other definitions.} The \emph{limit set} of a solution $\bx(\cdot)$ of a given dynamics is the set of accumulation points of $\bx(t)$ as $t \to +\infty$. 
A pure strategy $i$ \emph{belongs to the support of a Nash equilibrium} of a symmetric bimatrix game if  there is a Nash equilibrium $(\bx, \by)$ 
such that $x_i>0$ (or equivalently, due to the symmetry of the game, a Nash equilibrium $(\bx, \by)$  such that $y_i>0$). 
Finally, the pure strategy $i$ is \emph{eliminated}  (for a given solution $\bx(\cdot)$ of a given dynamics) if $x_{i}(t) \to 0$ as $t \to + \infty$. 

We show that there are games with a unique Nash equilibrium but such that, under the best-reply 
dynamics and the replicator dynamics, all strategies in the support of this equilibrium are eliminated for almost all initial conditions.

\section{Best-reply dynamics}
\label{sec:BR}
\subsection{A reminder on Rock-Paper-Scissors}


A general Rock-Paper-Scissors game (RPS) is a $3 \times 3$ symmetric game with payoff matrix %
\begin{equation}
\label{eq:RPS-gen} %
\left(\begin{array}{ccc}
 a_1   & b_2 & c_3   \\
 c_1   & a_2 & b_3   \\
 b_1   & c_2 &  a_3  \\
\end{array}\right) \quad \mbox{ with } 
b_i < a_i < c_i, \, 
i=1,2,3.
\end{equation}
%
(As the game is symmetric, we only indicate the payoffs of the row player.) 
These games have a unique Nash equilibrium. It is symmetric and completely mixed. We say that the game is \emph{outward cycling} if 
\begin{equation}
\label{eq:outward-RPS} 
 \prod_{i=1}^3 (a_i - b_i) > \prod_{i=1}^3 (c_i - a_i)
\end{equation}
In that case, almost all solutions of the best-reply dynamics converge to a triangle, which Gaunersdorfer and Hofbauer (1995) called the \emph{Shapley triangle} after Shapley (1964).  It is defined by
\begin{equation}
\label{eq:ST}
ST=\left\{\bx \in S_3: V(x)=0\right\} 
\mbox{ with } V(\bx)=\max_{1 \leq i \leq 3}(\bU\bx)_i -\sum_{1 \leq i \leq 3} a_i x_i 
\end{equation}
\begin{prop}[Gaunersdorfer and Hofbauer, 1995] 
\label{prop:BR-RPS}
In an outward cycling RPS game, for every initial condition different from the equilibrium, the solution of the best-reply dynamics is uniquely defined and its limit set is the Shapley triangle  (\ref{eq:ST}). \footnote{If $\prod_{i=1}^3 (a_i - b_i) = \prod_{i=1}^3 (c_i - a_i)$, e.g., if the game is zero-sum, the Shapley triangle is degenerate and coincides with the equilibrium; if $\prod_{i=1}^3 (a_i - b_i) < \prod_{i=1}^3 (c_i - a_i)$, the Shapley triangle is empty. In both cases, all solutions of the best-reply dynamics converge to the equilibrium.}\end{prop}
A RPS game has \emph{cyclic symmetry} if the payoffs $a_i$, $b_i$, $c_i$ are independent of $i$. The Nash equilibrium is then $(1/3, 1/3, 1/3)$ and, up to a rescaling that does not affect the equilibrium nor the dynamics we study, the payoffs may be taken of the form:
\begin{equation}
\label{eq:RPS-sym} %
\left(\begin{array}{ccc}
 0   & -\alpha & \beta \\
 \beta   & 0 & -\alpha \\
 -\alpha   & \beta & 0 \\
\end{array}\right) \quad \mbox{ with }\alpha>0, \beta>0.
\end{equation}
The outward cycling condition (\ref{eq:outward-RPS}) then boils down to $\alpha> \beta$, and the Shapley triangle (\ref{eq:ST}) to
\begin{equation}
\label{eq:STsym}
ST=\left\{\bx \in S_3 : \max_{1 \leq i \leq 3}(\bU\bx)_i=0\right\}. 
\end{equation}
We now describe in detail the behavior of the best-reply dynamics in  RPS games, and give a sketch of proof of Proposition \ref{prop:BR-RPS}, as this allows to introduce some crucial tools. The first one is a version of the improvement principle (Monderer and Sela, 1997). It says that when the solution of the best-reply dynamics points towards a pure best-reply $i$, only certain strategies can become best-replies: those that are better replies to $i$ than $i$ itself. 
\begin{lem}[Improvement principle] 
\label{lm:ip}
Let $\bx(\cdot)$ be a solution of the best-reply dynamics. 
Assume that on the interval $]T, T'[$, with $T<T'$, the unique best-reply to $\bx(t)$ is strategy $i$. 
If strategy $j \neq i$ is a best-reply to $\bx(T')$ then $u_{ji} > u_{ii}$.
\end{lem}
\begin{proof}
See Viossat (2008, Lemma 4.2). 
\end{proof}
Assume for instance that in a RPS game,  strategy 1 is currently the unique best-reply to the population profile $\bx(t)$, so that the solution points towards $\be_1$; that is, $\dot \bx=\be_1 - \bx$. Since $(\be_1, \be_1)$ is not a Nash equilibrium, a new best-reply must arise. By the improvement principle, this can only be strategy 2. The solution then points towards the edge $\be_1- \be_2$. Since in the game restricted to strategies $1$ and $2$, strategy $2$ strictly dominates strategy $1$, strategy $2$ immediately becomes the unique best-reply. Therefore the solution points towards $\be_2$, then towards $\be_3$, then towards $\be_1$ again,...

By itself, this cyclic behaviour does not preclude convergence to equilibrium. Actually, if  $\prod_{i=1}^3(a_i - b_i) < \prod_{i=1}^3(c_i - a_i)$ then the solutions cycle inwards, and the times at which their direction change accumulate as $\bx(t)$ converges, in finite time, to the equilibrium (Gaunersdorfer and Hofbauer, 1995). 

However, for outward cycling RPS games 
and for solutions that do not start at the equilibrium, this cyclic behavior goes on forever. 
This follows  from the following observations, which we do not prove. Below, the function $V$ is defined in (\ref{eq:ST}) and $v(t)=V(\bx(t))$. 

(i) If the game is outward cycling,  then $V(\bx)$ is zero on the Shapley triangle, positive outside it, and negative inside, with its unique minimum attained at the equilibrium point. 

(ii) When the solution points towards  a pure strategy (that is, $\dot \bx=\be_i - \bx$ for some $i$), then $\dot v(t)=-v(t)$. 


Consider a solution that does not start at the equilibrium. Combining (i), (ii), and the above described cyclic behavior, we get that the solution cannot approach the equilibrium, therefore the times at which the direction changes cannot accumulate and the cyclic behavior goes on for ever; thus, by (ii), $v(t) \to 0$ hence $\bx(t) \to ST$. The limit set of the solution is then easily seen to be the whole triangle.  

\subsection{A $6 \times 6$ game} 

Consider the following $6 \times 6$ symmetric game: 
\begin{equation}
\label{eq:66game} \left(\begin{array}{ccc|ccc}
 0  & -3 & 1  & -1 & -1 & -1 \\
 1  & 0  & -3 & -1 & -1 & -1 \\
 -3 & 1  & 0  & -1 & -1 & -1 \\
\hline
 -4 & -4 & 3  & 0  & -5 & 1  \\
 -1 & -1 & -3 & 1  & 0  & -5 \\
 -1 & -1 & -3 & -5 & 1  & 0  \\
\end{array}\right)
\end{equation}
Let $G_{123}$ and $G_{456}$ denote the $3 \times 3$ games obtained from (\ref{eq:66game}) by restricting the players to their three first 
and to their three last strategies, respectively. Both $G_{123}$ and $G_{456}$ are outward
cycling RPS games with cyclic symmetry. Their unique Nash
equilibrium correspond in the whole game to, respectively: 
%
%
$$\bn_{123}=\left(\frac{1}{3},\frac{1}{3},\frac{1}{3},0,0,0\right)  \quad \mbox{ and } \bn_{456}=\left(0,0,0,\frac{1}{3},\frac{1}{3},\frac{1}{3}\right)$$
The payoffs are chosen so that $(\bn_{123},\bn_{123})$ be a 
Nash equilibrium of (\ref{eq:66game}) but not $(\bn_{456},
\bn_{456})$. 
\begin{prop}
\label{prop:uniqueNash66}
The game (\ref{eq:66game}) has a unique Nash equilibrium:
$(\bn_{123},\bn_{123})$.
\end{prop}
\begin{proof}
See Appendix A. 
\end{proof}
Proposition \ref{prop:uniqueNash66} does not only state that $(\bn_{123},\bn_{123})$ is the unique symmetric Nash equilibrium,
but also that there are no asymmetric Nash equilibria. Nevertheless, from almost all initial
conditions, all strategies in its support are eliminated. More precisely, let $ST_{456}$ denote the Shapley triangle:
\begin{equation} \label{eq:def-ST-456}
ST_{456}=\left\{\bx: x_4 + x_5 +x_6=1 \mbox{ and } \max_{4 \leq i
\leq 6} (\bU\bx)_i=0\right\}
\end{equation}
\begin{prop}
\label{prop:66}
For almost every mixed strategy $\bx$ in $S_6$, there is a unique solution $\bx(\cdot)$ of (\ref{eq:defBR}) such that $\bx(0)=\bx$, and its limit set is the Shapley triangle $ST_{456}$. 
\end{prop}
%
\begin{proof} The proof relies on the improvement principle and the better-reply structure of the game, described in Fig. 1 below:
\unitlength 0,71mm
\begin{center}
\begin{picture}(175,80)
\put(-5,10){\line(3,5){36}}
\put(-5,10){\line(1,0){72}}
\put(31,70){\line(3,-5){36}}
\put(67,10){\line(1,0){42}}
\put(109,10){\line(3,5){36}}
\put(109,10){\line(1,0){72}}
\put(145, 70){\line(3,-5){36}}
\put(34, 10){\line(-1,1){5}}
\put(34, 10){\line(-1,-1){5}}
\put(91, 10){\line(-1,1){5}}
\put(91, 10){\line(-1,-1){5}}
\put(10, 35){\line(-1,6){1.2}}
\put(10, 35){\line(6,1){6}}
\put(49, 40){\line(-1,-3){2}}
\put(49, 40){\line(3,-1){5.2}}
%
\put(148, 10){\line(-1,1){5}}
\put(148, 10){\line(-1,-1){5}}
\put(124, 35){\line(-1,6){1.2}}
\put(124, 35){\line(6,1){6}}
\put(163, 40){\line(-1,-3){2}}
\put(163, 40){\line(3,-1){5.2}}
%
\put(30,74){$1$}
\put(144,74){$6$}
\put(-9,4){$2$}
\put(67,4){$3$}
\put(105,4){$4$}
\put(181,4){$5$}

%
%
 \end{picture}
\end{center}
\vspace{-0.5cm}
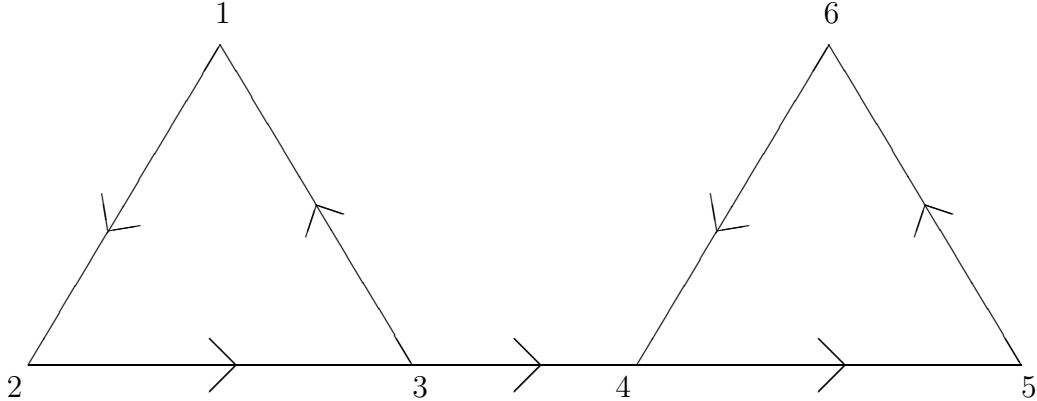
\begin{small} Figure 1: Better-replies to pure strategies in game (\ref{eq:66game}). An arrow from $i$ to $j$ means that $u_{ij}>u_{ii}$.\end{small}\\

Consider a solution $\bx(\cdot)$ of the best-reply dynamics. We may assume that there is a unique best-reply to $\bx=\bx(0)$, since this holds for almost all $\bx$ in  $S_6$. There are then two cases.
 
\emph{Case 1: the unique best-reply to $\bx(0)$ is strategy 4, 5 or 6. } Assume for concreteness that  this is strategy $4$. The improvement principle (Lemma \ref{lm:ip}) and the same reasoning as for RPS games imply that the solution first points towards $e_4$, then towards $e_5$, then towards $e_6$, then towards $e_4$ again, in a cyclic fashion. It may be shown exactly as in (Viossat, 2008, p.33) that the times at which the direction of the solution changes do not accumulate.\footnote{The idea is to show that the function $W(\bx)=\max_{4\leq i \neq  j \leq 6} [(\bU\bx)_i - (\bU\bx)_j]$ is bounded away from zero.} 
 It follows that this cyclic behavior goes on for ever. Therefore, strategies $1$, $2$ and $3$ never become best-replies, hence $x_i(t)=x_i(0)e^{-t} \to 0$ for all $i$ in $\{1,2,3\}$. Moreover, when strategy $i \in \{4, 5, 6\}$ is the unique best-reply, the function $v(t)= \max_{4 \leq i \leq 6} (\bU\bx)_i(t)$ is equal to $v(t)=(\bU\bx)_i(t)=\be_i \cdot \bU\bx(t)$ and satisfies: 
\begin{equation}
\dot v = \be_i \cdot \bU\dot \bx= \be_i \cdot \bU (\be_i - \bx) = - \be_i \cdot \bU\bx = - v
\end{equation} 
Therefore $v(t) \to 0$, hence $x(t) \to ST_{456}$ as $t \to +\infty$.  

\emph{Case 2: the unique best-reply to $\bx(0)$ is strategy 1, 2 or 3. } 
Assume for concreteness that  this is strategy $4$. If none of the strategies 4, 5 and  6 ever becomes a best-reply, the solution points towards $e_1$, then towards $e_2$, then towards $e_3$, etc., and due to the same reasoning as in case 1, its limit set will be the Shapley triangle 
$$ST_{123}=\left\{\bx: x_1 + x_2 +x_3=1 \mbox{ and } \max_{1 \leq i \leq 3} (\bU\bx)_i=0\right\}.$$   
This is impossible, because the payoffs are such that at one of the vertices of this triangle, the closest to $\be_3$, strategy 4 is the unique best-reply. This vertex is given by $\frac{1}{13} (1, 3, 9, 0, 0, 0)$, see Gaunersdofer and Hofbauer (1995, Eq. (3.6)).
 
Thus, there exists a first time $T>0$ at which one of the strategies $4$, $5$ and $6$ becomes a best-reply. Due to the improvement principle and to the better-reply structure of the game (Fig. 1), this can only be strategy $4$, and just before $T$, the unique best-reply was strategy $3$.

\noindent There are then two subcases:

\emph{Subcase 2.1}: the pure best-replies at time $T$ are strategies $1$, $3$ and $4$ (that is, strategies $1$ and $4$ become best-replies at the same time). The dynamics then admits several solutions and becomes more difficult to analyze. Fortunately, the solution can be precisely traced back in time (in backward time, starting from T, it moves away from $e_3$ along a straight line, then away from $e_2$,...). It follows that the set of initial conditions for which this case occurs is contained in the intersection of the simplex with a countable union of hyperplanes of $\R^N$, none of which contains the simplex. Therefore, this set has Lebesgue measure zero (with respect to the simplex) and we can neglect this case. 

\emph{Subcase 2.2}: the pure best-replies at time $T$ are strategies $3$ and $4$. The solution will then point towards the edge $e_3-e_4$. Since in the game reduced to strategies $3$ and $4$, strategy $4$ strictly dominates strategy $3$, it follows that strategy $4$ becomes the unique best-reply and we are back to case 1. 
\end{proof}

\emph{Robustness to perturbations of the payoffs. }  The above proof uses only strict inequalities, which are unaffected by sufficiently small perturbations of the payoffs (the only modification is that the Shapley triangles and the underlying functions V must be defined as in (\ref{eq:ST}) because the diagonal terms need no longer be zero). Moreover,  since the game is a bimatrix game with a unique Nash equilibrium, it follows that any game in its neighborhood has a unique Nash equilibrium, and with the same support (Jansen, 1981). Therefore: 
\begin{prop}
\label{prop:BR-rob}
There exists a neighborhood of game (\ref{eq:66game}) such that, for any symmetric game in this neighborhood, the unique Nash equilibrium has support in $\{1,2,3\} \times \{1,2,3\}$, but for almost all initial conditions, strategies $1$, $2$ and $3$ are eliminated by the best-reply dynamics. 
\end{prop}

\section{Replicator Dynamics}
\label{sec:REP}
Up to a further rescaling,  the payoff matrix of an outward cycling RPS game with cyclic symmetry (\ref{eq:RPS-sym}) may be taken of the form:
\begin{equation}
\label{eq:RPS-symbis} %
\left(\begin{array}{ccc}
 0   & -1 & \eps \\
 \eps   & 0 & -1 \\
 -1   & \eps & 0 \\
\end{array}\right)
\mbox{ with } 0<\eps < 1 
\end{equation}
The behavior of the replicator dynamics in such games is well known. The boundary $\Gamma=\{\bx \in S_3 : x_1 x_2 x_3=0\}$ forms a heteroclinic cycle, that is, a globally invariant set consisting of saddle rest-points and saddle orbits connecting these rest-points. Moreover:
\begin{prop}
\label{prop:REP}[Zeeman, 1980
; Gaunersdorfer and Hofbauer, 1995] 
In game $(\ref{eq:RPS-symbis})$, the set $\Gamma$ is asymptotically stable, all interior solutions that do not start at the equilibrium $(1/3, 1/3, 1/3)$ converge to $\Gamma$ and the limit set of their time-average is the Shapley triangle $(\ref{eq:STsym})$. 
\end{prop}
\noindent (If $\by(\cdot)$ is a solution of (\ref{eq:defREP}), its time-average at $t \neq 0$ is $\frac{1}{t} \int_0^t \by(s) \, ds$.)
  
Two other facts will prove useful: first, in game (\ref{eq:RPS-symbis}), the mean payoff is always nonpositive: 
\begin{lem}
\label{lm:bad}
Let $\bU$ denote the payoff matrix (\ref{eq:RPS-symbis}): $\forall \bx \in S_3, \bx \cdot \bU \bx \leq 0$
\end{lem} 
\begin{proof}
A standard computation shows that $\bx \cdot \bU \bx= \frac{(-1+\eps)}{2}\left[1-\sum_{i=1}^3 x_i^2\right]$ which is nonnegative since $\bx \in S_3$.  
\end{proof} 

Second, as computed by Gaunersdorfer and Hofbauer (1995, Eq. (3.6)), the vertex of the Shapley triangle 
closest to $\be_3$ is given by 
\begin{equation}
\label{eq:qbis}
\bar{\bq}=\frac{1}{1+\eps+\eps^2}(\eps^2, \eps, 1)
\end{equation} 

Consider a solution $\by(\cdot)$ of the replicator dynamics that does not start at the equilibrium. Proposition \ref{prop:REP} implies that $\bar{\bq}$ is an accumulation point of the time-average of $\by(t)$. Moreover, for $\eps$ small enough ($\eps < 1/4$ suffices), $4\bar{q}_3 - 3>0$; this implies the following result: 
\begin{equation}
\label{eq:limsup}
\limsup \int_0^t (4y_3(s) - 3) ds = +\infty
\end{equation}

Now consider the following $7 \times 7$ symmetric game:
\begin{equation}
\label{eq:77game}%
 \left(\begin{array}{ccc|c|ccc}
    0     &    -1    & \eps & -10 &  -1/3 + \eps & -1/3 + \eps & -1/3 +\eps  \\
 \eps &    0     &    -1    & -10 &  -1/3 + \eps & -1/3 + \eps & -1/3 +\eps  \\
    -1    & \eps &    0     & -10 &   -1/3 + \eps & -1/3 + \eps & -1/3 +\eps \\
\hline
  -2     &    -2    &     2      &  0  & -1/3       &      -1/3   &   -1/3   \\
\hline
  -1/3   &   -1/3   &   -1/3    & 10  &   0     &    -1    &    \eps     \\
  -1/3   &   -1/3   &   -1/3    & 10  &    \eps     &        0        &       -1         \\
 -1/3   &   -1/3   &   -1/3     & 10  &       -1        &    \eps     &        0         \\
\end{array}\right)
\end{equation}
%
%
%
%
with $\eps>0$ small enough.\footnote{In the proofs, for simplicity, we use $\eps<1/48$, but the results extend easily to $\eps<1/6$, and probably beyond.}  
The games obtained by restricting both players to their three first or to their three last strategies are outward cycling Rock-Paper-Scissors games with cyclic symmetry. The Nash equilibria of these games correspond in the whole game to rest points of the replicator dynamics, which we denote by $\bn_{123}$ and $\bn_{567}$: 
$$\bn_{123}=\left(\frac{1}{3},\frac{1}{3},\frac{1}{3},0,0,0,0\right)  \quad ; \quad \bn_{567}=\left(0,0,0,0,\frac{1}{3},\frac{1}{3},\frac{1}{3}\right)$$
The heteroclinic cycles of the RPS games correspond to heteroclinic cycles of the whole game, which we denote by $\Gamma_{123}$ and $\Gamma_{567}$: 
$$\Gamma_{123}=\{ \bx \in S_7 : x_1 + x_2 + x_3=1 \mbox{ and } x_1 x_2 x_3 =0 \};$$ 
$$\Gamma_{567}=\{\bx \in S_7 : x_5 +x_6 +x_7=1 \mbox{ and } x_5 x_6 x_7=0  \}.$$
\begin{prop}
\label{prop:uniqueNash77}
$(\bn_{123},\bn_{123})$ is the unique Nash equilibrium of game $(\ref{eq:77game})$\footnote{There are no asymmetric
Nash equilibria.}
\end{prop}
\begin{proof}
See Appendix A 
\end{proof}
In spite of Proposition \ref{prop:uniqueNash77}, the heteroclinic cycle $\Gamma_{123}$ is not asymptotically stable. Indeed, at $\be_3$, the unique best-reply is strategy $4$.  By contrast, though $(\bn_{567}, \bn_{567})$ is not an equilibrium of (\ref{eq:77game}): 
\begin{prop}
\label{prop:Gamma567}
The heteroclinic cycle $\Gamma_{567}$ is asymptotically stable. 
\end{prop}
\begin{proof}
$\Gamma_{567}$ is asymptotically stable on the face spanned by $\be_5$, $\be_6$, $\be_7$ due to Proposition \ref{prop:REP}. Moreover, near the vertices $\be_5$ to $\be_7$, the  payoffs of strategies $1$ to $4$ are less than the mean payoff, hence the shares of strategies $4$ to $7$ decrease. Then apply Thm. 17.5.1 of Hofbauer and Sigmund (1998). 
\end{proof}  
Thus, if a solution 
of the replicator dynamics approaches $\Gamma_{567}$ arbitrarily closely, then it converges to it. We will show that this occurs for almost all initial conditions. 
Together with Proposition \ref{prop:uniqueNash77}, this implies that for almost all initial conditions, all pure strategies in the support of the unique equilibrium of game (\ref{eq:77game}) are eliminated. 

Roughly, if the solution starts close to the equilibrium, then it first spirals towards the heteroclinic cycle $\Gamma_{123}$. Eventually, it spends enough time close to $\be_3$, where the unique best-reply is strategy $4$, for $x_4$ to increase substantially. Since strategies $5$, $6$, and $7$ have very good payoffs again strategy $4$, this triggers a subsequent  increase in $x_5$, $x_6$ and $x_7$. The solution then cycles towards $\Gamma_{567}$. However, $x_4$ then decreases, which may lead to a come-back of strategies $1$, $2$, $3$, and the whole process might start again. The difficulty is to make sure that, each time this process runs, the solution gets closer to $\Gamma_{567}$.  

For the replicator dynamic, this can be shown due to the last important property of game (\ref{eq:77game}): against strategies $4$ to $7$, strategies $1$ to $3$ have the same payoffs. That is, for any $i$, $i'$ in $\{1,2,3\}$ and any $j$ in $\{4,5,6,7\}$, $u_{ij}=u_{i'j}$. Similarly, against strategies $1$ to $4$, strategies $5$ to $7$ have the same payoffs.  Due to linearity properties of the replicator dynamics, this implies that the dynamics may be decomposed as we now explain. 

Let $\bx(\cdot)$ be an interior solution of the replicator dynamics. For each
$i$ in $\{1,2,3\}$, define $\bar{x}_i(t)$ as the
share of strategy $i$ at time $t$ relative to the total share of
strategies $1$, $2$ and $3$:
\begin{equation}
\label{eq:REP-def-bar} \bar{x}_i:=\frac{x_i}{x_1 + x_2 +x_3}
\end{equation}
and let $\bar{\bx}=(\bar{x}_1,\bar{x}_2,\bar{x}_3)$. For $i \in \{5, 6, 7\}$, define similarly: %
\begin{equation} \label{eq:REP-def-hat}
\hat{x}_i:=\frac{x_i}{x_5 + x_6 +x_7}
\end{equation}
and let $\hat{\bx}=(\hat{x}_5,\hat{x}_6,\hat{x}_7)$. Finally, let
$$\lambda(t)=x_1(t) +x_2(t) +x_3(t)  \mbox{ and } \mu(t)=x_5(t) +x_6(t) +x_7(t)$$ 
denote respectively the total share of the three first and of the three last strategies at time $t$. 
The evolution of $\bx$ is fully described by the joint evolution of
$\bar{\bx}$, $\hat{\bx}$, $\lambda$ and $\mu$. The interest of this description is that, up to a change in velocity, $\bar{\bx}$ and $\hat{\bx}$ follow the
replicator dynamics of the Rock-Paper-Scissors game (\ref{eq:RPS-symbis}). 

Formally, let $\bar{\tau}(t)$ denote the rescaled time%
\begin{equation}
\label{eq:taubar}
\bar{\tau}(t):=\int_0^t \lambda(s)ds
\end{equation}
Let both $\bar \bU$ and $\hat \bU$ denote the payoff matrix (\ref{eq:RPS-symbis}), depending on whether it arises as the top-left or the bottom-right corner of game (\ref{eq:77game}).\footnote{The top-left and bottom-right RPS games of (\ref{eq:77game}) need not be the the same for the results to hold, this is just to minimize the number of parameters.} 
\begin{lem}
\label{lm:rescale} 
Let $\by(\cdot)$ denote the solution of (\ref{eq:defREP}) 
in the RPS game \emph{(\ref{eq:RPS-symbis})}, with initial condition $\by(0)=\bar{\bx}(0)$. We have:
\begin{equation}\label{eq:decomp} 
\dot{\bar{x}}_i=\lambda \bar{x}_i\left[(\bar{\bU}\bar{\bx})_i - \bar{\bx} \cdot
\bar{\bU}\bar{\bx}\right] \hspace{0.5 cm} \forall i=1,2,3
\end{equation}
\begin{equation}
\label{eq:decomp2}\forall t \in \R, \hspace{0.5 cm} \bar{\bx}(t)=\by(\bar{\tau}(t)) \end{equation}
 \end{lem}
\begin{proof}%
The proof of (\ref{eq:decomp}) is the same as the proof of Lemma 5.2 of Viossat (2007). Due to (\ref{eq:decomp}), $\by(\bar{\tau}(t))$ and $\bar{\bx}(t)$ are
solutions of the same differential equation, which admits a unique solution through each initial condition. This proves (\ref{eq:decomp2}). 
\end{proof}
%
Similarly, if $\bz(\cdot)$ is the solution of the replicator dynamics in game (\ref{eq:RPS-symbis}) with initial condition $\bz(0)=\hat{\bx}(0)$, and $\hat{\tau}(t)$ is the rescaled time 
\begin{equation}
\label{eq:tauhat}
\hat{\tau}(t) = \int_0^t \mu(s)ds
\end{equation}
then 
\begin{equation}
\label{eq:rescale}
\forall t \in \R, \hspace{0.5 cm} \hat{\bx}(t)=\bz(\hat{\tau}(t))
 \end{equation}


We are now ready to prove the main result of this section:
\begin{prop}
\label{prop:77} For any interior initial condition
$\bx=\bx(0)$ such that neither $x_1=x_2=x_3$
nor $x_5=x_6=x_7$, the solution 
of the replicator dynamics converges to $\Gamma_{567}$. In particular, all pure strategies in the support of the unique equilibrium of game (\ref{eq:77game}) are eliminated.
\end{prop}

\begin{proof}
The assumptions imply that $\bar{\bx}(0)$ and $\hat{\bx}(0)$ are well defined and different from
$(\frac{1}{3},\frac{1}{3},\frac{1}{3})$. We must show that $\bx(t)$ converges to $\Gamma_{567}$.  

If $\lambda(t) \to 0$, then $\bx(t)$ converges to the face $\lambda=0$.  Since on this face the payoff of strategy $4$ is strictly smaller than the payoff of $\bn_{567}$, standard, domination-like arguments imply that $x_4(t) \to 0$, hence $\mu(t) \to 1$ (see, e.g., Samuelson and Zhang, 1992). Due to (\ref{eq:rescale}) and Proposition \ref{prop:REP}, this implies that $\bx(t) \to \Gamma_{567}$ and we are done. Thus it suffices to show that $\lambda(t)$ 
converges to zero. 


Assume by contradiction that this is not the case. 

\begin{cl} \label{cl:1} 
Recall that $\bar{\tau}(t)= \int_0^t \lambda(s)ds$. We have: $\bar{\tau}(t)  \to +\infty$ as $t \to +\infty$. 
\end{cl}
\begin{proof}
$\lambda(t) \nrightarrow 0$ and $\lambda$ is clearly Lipschitz.
\end{proof}
\begin{cl} \label{cl:2} $\limsup_{t \to +\infty} \int_0^t \left[4x_3(s) - 3\lambda(s)\right] \, ds = +\infty$ 
\end{cl}
\begin{proof}
By  (\ref{eq:decomp2}), definition of 
$\bar{\tau}$, and a change of variable,
$$\int_0^t \lambda(s) \bar{\bx}(s) \, ds = \int_0^t \dot{\bar{\tau}}(s) \by(\bar{\tau}(s))  \, ds= \int_0^{\bar{\tau}(t)} \by(s) \, ds.$$
Therefore
$$\int_0^t (4x_3(s) - 3\lambda(s)) \, ds= \int_0^{t} \lambda(s) (4\bar{x}_3(s) - 3) \, ds = \int_0^{\bar{\tau}(t)} (4y_3(s) - 3) \, ds$$
Since $\by(0) \neq (1/3, 1/3, 1/3)$, the result follows from Claim \ref{cl:1} and Eq. (\ref{eq:limsup}). 
\end{proof}
\begin{cl} \label{cl:3} 
$\displaystyle \limsup_{t \to +\infty} \mu(t) \geq
\frac{1}{1+\eps}$
\end{cl}
\begin{proof}
Using (\ref{eq:defREP}), an easy computation shows that 
\begin{equation}
\label{eq:lambdax4ante}%
\frac{d}{dt}\ln\left(\frac{x_4}{\lambda}\right)= 4x_3 + 10 x_4 - 2\lambda
- \lambda \bar{\bx} \cdot \bar{\bU}\bar{\bx} -\eps\mu  \geq 4x_3   + 10 x_4 - 2\lambda -\eps\mu
\end{equation}  
where the inequality follows from Lemma \ref{lm:bad}. 

Assume by contradiction that $\limsup_{t \to +\infty} \mu(t) < \frac{1}{1+\eps}$. Thus, omitting time arguments, there exists a time $T$ such that %
for all $t \geq T$, $(1+\eps)\mu < 1= \mu + \lambda +
x_4$ hence $\eps \mu \leq \lambda + x_4$. Together
with (\ref{eq:lambdax4ante}), this implies that for $t \geq T$:
\begin{equation}
\label{eq:lambdax4}%
\frac{d}{dt}\ln\left(\frac{x_4}{\lambda}\right) \geq 4 x_3 + 9 x_4 - 3\lambda  \geq 4x_3 - 3\lambda%
\end{equation}
By Claim \ref{cl:2}, it follows that $\limsup \ln(x_4/\lambda)=+\infty$. Thus, there exists a time $T'>T$ such that $x_4 \geq \lambda$. Due to the first inequality in (\ref{eq:lambdax4}), for $t \geq T'$, $x_4$ remains greater than $\lambda$ and $\frac{d}{dt}\ln\left(\frac{x_4}{\lambda}\right) \geq 9x_4 - 3\lambda \geq 6\lambda$. 
 %
%
By Claim \ref{cl:1}, this implies that $x_4/\lambda \to +\infty$ hence $\lambda \to 0$, a contradiction.
\end{proof}

We now conclude. Recall the definition of $\hat{\tau}$ in (\ref{eq:tauhat}). A corollary of Claim \ref{cl:3} is that 
$\hat{\tau}(t) 
\to +\infty$ as $t \to +\infty$. By (\ref{eq:rescale}) and Proposition \ref{prop:REP}, it follows that $\hat{\bx}$
converges to the heteroclinic cycle of game (\ref{eq:RPS-symbis}). It is easy to check that along this cycle, the mean payoff is always greater than $-\frac{1}{4}$. Therefore:
\begin{equation}
\label{eq:hautleshat} %
\exists T_1 \geq 0,\, \forall t \geq T_1,\, \hat{\bx}(t) \cdot \hat{\bU}\hat{\bx}(t) \geq -\frac{1}{4} - \eps%
\end{equation}%
Moreover, (\ref{eq:defREP}) and a somewhat tedious computation show that: 
\begin{equation}
\label{eq:lambdamu}%
\frac{d}{dt}\ln\left(\frac{\mu}{\lambda}\right)%
=\mu \, \left(\hat{\bx}\cdot\hat{\bU}\hat{\bx}
+\frac{1}{3}-\eps\right)- \lambda.\left(\frac{1}{3}+\bar{\bx} \cdot \bar{\bU}\bar{\bx}\right) + 20 x_4 %
\end{equation}
Assuming $\eps \leq \frac{1}{48}$, (\ref{eq:hautleshat}), (\ref{eq:lambdamu}) and Lemma \ref{lm:bad} imply that
for $t \geq T_1$:
\begin{equation}
\label{eq:lambdamubis}%
\frac{d}{dt}\ln\left(\frac{\mu}{\lambda}\right)\geq\frac{\mu}{24} -\frac{\lambda}{3}%
\end{equation}
It follows from Claim \ref{cl:3} 
that there exists a time
$T_2 \geq T_1$ at which the ratio $\mu/\lambda$ is greater than
$16$. By (\ref{eq:lambdamubis}), this ratio then keeps increasing hence, by (\ref{eq:lambdamubis}) again, 
\begin{equation}
\forall t \geq T_2,\hspace{0.5 cm}
\frac{d}{dt}\ln\left(\frac{\mu}{\lambda}\right)(t)\geq
\frac{16\lambda}{24} -\frac{\lambda}{3} \geq \frac{\lambda(t)}{3}
\end{equation}
By Claim \ref{cl:1}, this implies that $\lambda$ goes to
zero, a final contradiction.
\end{proof}

\emph{Perturbation of payoffs. } 
As for game (\ref{eq:66game}), any game sufficiently close to game (\ref{eq:77game}) in the payoff space has a unique Nash equilibrium, and its 
support is $\{1, 2, 3\} \times \{1, 2, 3\}$. We conjecture that the result of Proposition \ref{prop:77} generalizes to such nearby games. That is, for almost all initial conditions, the solution of the replicator dynamics converges to the boundary of the face spanned by $\be_5$, $\be_6$ and $\be_7$, hence all pure strategies in the support of the unique Nash equilibrium are eliminated. Our proof does not go through however, because Lemma \ref{lm:rescale} requires a very specific payoff structure.

\emph{Correlated equilibrium. } By contrast with the games of Viossat (2007, 2008), the Nash equilibrium of games (\ref{eq:66game}) and (\ref{eq:77game}) is not the unique correlated equilibrium. Whether reasonable dynamics may eliminate all strategies used in correlated equilibrium for almost all initial conditions is an open question.  

\emph{Other dynamics. }  A variant of Lemma \ref{lm:rescale} holds for the discrete-time replicator dynamics:
\begin{equation}\label{eq:drep}
x_i(n+1)=x_i(n) \frac{C+ (\bU\bx)_i}{C + \bx \cdot \bU\bx} \mbox{ with } C > - \min_{i,\mathbf{x}} (\bU\bx)_i 
\end{equation}
Thus, extending Proposition \ref{prop:77} to (\ref{eq:drep}) should be relatively simple. 
Proposition \ref{prop:77} might also extend to some classes of payoff functionnal dynamics 
\begin{equation}
\label{eq:pf}
\dot{x}_i= x_i \left[f([\bU \bx]_i) - \sum_j x_j f([\bU\bx]_j) \right] 
\end{equation}
and $f$ an increasing and sufficiently smooth function from $\R$ to $\R$. This might be hard to prove though, as Lemma \ref{lm:rescale} builds on linearity properties which are specific of the replicator dynamics.\footnote{One reason to hope for a generalization is that, in Rock-Paper-Scissors games, close to the equilibrium, dynamics (\ref{eq:pf}) behave as the replicator dynamics (Hofbauer and Sigmund, 1998, exercice 8.1.1; Viossat, 2011, footnote 6). }

Finally, there is a strong link between the best-reply dynamics and the time-average of the replicator dynamics (Gaunersdorfer and Hofbauer, 1995; Hofbauer et al., 2009). For this reason, 
we conjecture that Proposition \ref{prop:66} extends to (\ref{eq:defREP}); that is, in game (\ref{eq:66game}), for almost all initial conditions, all strategies in the support of the equilibrium are eliminated under (\ref{eq:defREP}). What we can show, in the same spirit, is that Proposition \ref{prop:77} extends to the best-reply dynamics, up to replacement of the heteroclinic cycle $\Gamma_{567}$ by the corresponding Shapley triangle: $$ST_{567}:=\left\{\bx \in S_7: x_5 +x_6 +x_7=1 \mbox{ and } \max_{5 \leq i \leq 7} (\bU\bx)_i=0\right\}$$
\begin{prop} \label{prop:universal2BR}
Assume that $0 < \eps < 2/9$. For any initial condition $\bx$
such that neither $x_1=x_2=x_3$ nor $x_5=x_6=x_7$, all solutions
of the best-reply dynamics converge to the Shapley triangle
$ST_{567}$.
\end{prop}
\begin{proof} See Appendix \ref{app:BR}. Compared to Proposition \ref{prop:66}, the added difficulty is to deal with initial conditions through which several solutions of (\ref{eq:defBR}) exist. This can be done due to a decomposition of the best-reply dynamics similar to Lemma \ref{lm:rescale}.\end{proof}

\section{Discussion}
\label{sec:discussion} In game (\ref{eq:77game}), the Nash
equilibrium is unique and quasi-strict, and therefore persistent,
regular, 
hence strongly stable, essential, strictly proper, strictly perfect, etc. (van Damme, 1991)
Thus, from the traditional, rationalistic point of view, it may be seen as the
unambiguous solution of the game. However, under two of the most studied dynamics, 
all strategies in the support of this Nash equilibrium are eliminated from almost all
initial conditions. This indicates an even wider gap between
strategic and evolutionary considerations that had been noted
before.

We conjecture that elimination of all strategies in the
support of Nash equilibria from almost all initial conditions occurs for many other dynamics, 
including multi-population dynamics. 
However, this might be hard to prove because this can only arise in relatively large games, in which having a precise understanding of dynamics more complex than the replicator dynamics or the best-reply dynamics might prove difficult. A way forward might be to consider nonlinear games and to replace, in the construction, Rock-Paper-Scissors games by hypnodisk games (Hofbauer and Sandholm, 2011).



 \begin{appendix}
 
 \section{Equilibrium uniqueness}
 \label{app:A}
 
In this section, we show that games (\ref{eq:66game}) and  (\ref{eq:77game}) have a unique equilibrium. We begin with a lemma used in both proofs. 

Consider a symmetric bimatrix game with pure strategy set $I=\{1,2,...,N\}$ and payoff matrix $\bU$. Let $I' \subset I$. For any $\bx$ in
$S_N$, define $\bx' \in \R_+^N$ by $x'_i=x_i$ if $i \in I'$ and $x'_i=0$ otherwise. Let $x(I')=\sum_{i \in I'} x_i$.

\begin{lem}
\label{lm:uniqueNash-lemgen} Let $(\bx,\by)$ be a Nash
equilibrium such that $x(I')y(I')>0$. Assume that against $\bx-\bx'$ and $\by-\by'$, the payoffs of a strategy $i$ in $I'$ is independent of $i$. That is, for all $i$ and $j$ in $I'$,
\begin{equation}
\label{eq:uniqueNash-lemgen} 
[\bU(\by-\by')]_i=[\bU(\by-\by')]_j \mbox{ and } [\bU(\bx-\bx')]_i=[\bU(\bx-\bx')]_j
\end{equation}
Then $(\bx',\by')$ induces an unnormalized Nash equilibrium of
the game restricted to $I' \times I'$. That is, for all $i$, $j$ in $I'$:
\begin{equation} \label{eq:induce} x'_i>0 \Rightarrow (\bU\by')_i \geq (\bU\by')_j \quad \mbox { and } \quad y'_i>0 \Rightarrow (\bU\bx')_i \geq (\bU\bx')_j \end{equation}
\end{lem}
\begin{proof}
Let $i \in I'$. If $x'_i>0$ then $x_i>0$, hence strategy $i$ is a best-reply to $\by$. Together with (\ref{eq:uniqueNash-lemgen}) this implies that for all $j$ in $I'$, 
$(\bU\by')_i-(\bU\by')_j= (\bU\by)_i-(\bU\by)_j \geq 0$. This proves the first part of (\ref{eq:induce}). The second part is symmetric.
\end{proof}

\noindent \textbf{Proof of proposition \ref{prop:uniqueNash66}.} 
Let $(\bx,\by)$ be a Nash equilibrium of (\ref{eq:66game}). We want to show that $\bx=\by=\bn_{123}=(1/3, 1/3, 1/3, 0, 0, 0)$. 

\emph{Step 1. } $x_4x_5x_6=0$ and by symmetry $y_4y_5y_6=0$.\\
Indeed, if $x_4x_5x_6>0$, then strategies $4$, $5$ and $6$ are all
best replies to $\by$, hence so is $\bn_{456}$. This cannot be because, as is easily checked, 
$\bn_{456}$ is strictly dominated by $\bn_{123}$.

\emph{Step 2. } $y_1+y_2+y_3>0$ and by symmetry $x_1 + x_2 + x_3>0$.\\
Assume by contradiction that $y_1=y_2=y_3=0$. It follows that
\begin{equation}
\label{eq:proofNash66-point-iv} \forall i \in \{1,2,3\},
(\bU\by)_i=-1<0
\end{equation}
Furthermore, due to Step 1, $\by$ has support in $\{4,5\}$,
$\{5,6\}$ or $\{4,6\}$. In any case, there exists $i$ in
$\{4,5,6\}$ such that $(\bU\by)_i\geq 0$. Together with
(\ref{eq:proofNash66-point-iv}), this implies that strategies $1$,
$2$ and $3$ are not best replies to $\by$, hence
$x_1=x_2=x_3=0$. Thus, both $\bx$
and $\by$ have support in $\{4,5,6\}$, hence $(\bx,\by)$ induces a Nash
equilibrium of the game restricted to $\{4,5,6\} \times
\{4,5,6\}$. This implies that $\bx=\by=\bn_{456}$, which
contradicts Step 1. 

\emph{Step 3. } $x_1=x_2=x_3$ and $y_1=y_2=y_3$.\\
Let $\bx_{123}=(x_1,x_2,x_3,0,0,0)$ and
$\bx_{456}=\bx-\bx_{123}=(0,0,0,x_4,x_5,x_6)$. Define $\by_{123}$ and $\by_{456}$
symmetrically. For every $i$ and $j$ in $\{1,2,3\}$, we have $(\bU\bx_{456})_i=(\bU\bx_{456})_j$ and $(\bU\by_{456})_i=(\bU\by_{456})_j$. 
Therefore, if follows from Step 2 and from Lemma \ref{lm:uniqueNash-lemgen} applied with $I'=\{1, 2, 3\}$ that 
$(\bx_{123}, \by_{123})$ is an unnormalized Nash equilibrium of the game restricted to $\{1,2,3\} \times \{1,2,3\}$. Therefore $\bx$ and $\by$ are both proportional to
$\bn_{123}$. 

\emph{Step 4. } $x_4 + x_5 +x_6=0$ and by symmetry $y_4 + y_5 + y_6=0$.\\
Assume by contradiction that $x_4 + x_5 + x_6>0$. Against $\bn_{123}$, every strategy $i$ in $\{4,5,6\}$ earns the same payoff: $-5/3$. Thus, by Step 3, for every $i$ and $j$ in $\{4,5,6\}$, we have
$(\bU\bx_{123})_i=(\bU\bx_{123})_j$ and $(\bU\by_{123})_i=(\bU\by_{123})_j$.
Together with lemma \ref{lm:uniqueNash-lemgen} with $I'=\{4,5,6\}$, this implies that if $y_4 + y_5 + y_6>0$ then $\bx_{456}$ and $\by_{456}$ are
proportional to $\bn_{456}$, hence $x_4x_5x_6>0$. This cannot be due to Step 1. Therefore,  $y_4 + y_5+ y_6=0$. But then, by Step 3, $\by=\bn_{123}$. Therefore, strategies $4$, $5$ and $6$ are not
best replies to $\by$. Therefore $x_4=x_5=x_6=0$.

Proposition \ref{prop:66} now follows from Steps 3 and 4. 
\finpreuve\\

\noindent \textbf{Proof of proposition \ref{prop:uniqueNash77}. } 
Recall the definition of $\bn_{123}$ and $\bn_{567}$: 
$$\bn_{123}=\left(\frac{1}{3},\frac{1}{3},\frac{1}{3},0,0,0,0\right) \quad ; \quad \bn_{567}=\left(0,0,0,0,\frac{1}{3},\frac{1}{3},\frac{1}{3}\right).$$
Let $(\bx,\by)$ be a Nash equilibrium of (\ref{eq:77game}). Consider the conditions:
\begin{equation}
\label{eq:cond1} x_1 +x_2+x_3>0 \mbox{ and } y_1 + y_2 + y_3>0
\end{equation}
\begin{equation}
\label{eq:cond2} x_5 +x_6+x_7>0 \mbox{ and } y_5 + y_6 + y_7>0
\end{equation}
Note that, due to Lemma \ref{lm:uniqueNash-lemgen}:
\begin{lem}
\label{lm:cond12} If $(\ref{eq:cond1})$ holds, then $x_1=x_2=x_3$ and $y_1=y_2=y_3$. If $(\ref{eq:cond2})$ holds, then $x_5=x_6=x_7$ and $y_5=y_6=y_7$. 
\end{lem}
\noindent Now examines 4 cases, depending on whether (\ref{eq:cond1}) and (\ref{eq:cond2}) hold or not: 

\emph{Case 1. If $(\ref{eq:cond1})$ holds. } Then, by lemma \ref{lm:cond12}, $y_1=y_2=y_3$. Therefore $\bn_{567} \cdot \bU \by > (\bU \by)_4$, hence $x_4=0$. By symmetry, $y_4=0$.

\emph{Subcase 1.1.  If furthermore $(\ref{eq:cond2})$ holds. } Then
by lemma \ref{lm:cond12}, $y_5=y_6=y_7$. Since $y_4=0$
and $y_1=y_2=y_3$, it follows that $\by$ is a convex combination
of $\bn_{123}$ and $\bn_{567}$. Against both of these strategies, the payoff of $\bn_{123}$ is strictly greater than the payoff of strategies $5$, $6$ and $7$. 
Thus, the latter cannot be best-replies to $\by$, hence $x_5 +x_6 +x_7=0$. This contradicts (\ref{eq:cond2}).

\emph{Subcase 1.2. If $(\ref{eq:cond2})$ does not hold. } Without
loss of generality, assume that $y_5 +y_6 +y_7=0$. Since  $y_4=0$
and $y_1=y_2=y_3$, this implies that $\by=\bn_{123}$. Therefore,
as above, none of the strategies $5$, $6$ and $7$ is a
best reply to $\by$. Therefore $x_5 +x_6 +x_7=0$ which by the
same argument implies $\bx=\bn_{123}$. Therefore,
$\bx=\by=\bn_{123}$.

\emph{Case 2. If $(\ref{eq:cond1})$ does not hold. } %
Without loss of generality, assume $x_1 +x_2 +x_3=0$. This implies
that $\bn_{567}$ is a strictly better response to $\bx$ than
strategy $4$. Thus, $y_4=0$.

\emph{Subcase 2.1. If furthermore (\ref{eq:cond2}) holds. } Then
$\by$ is a convex combination of $\bn_{567}$ and strategies
$1,2,3$. This implies that $\bn_{123}$ is a strictly better
response to $\by$ than either $5$, $6$ or $7$. Therefore,
$x_5=x_6=x_7=0$, contradicting (\ref{eq:cond2}).

\emph{Subcase 2.2. If $(\ref{eq:cond2})$ does not hold. } Then
$x_5+x_6 +x_7=0$ or $y_5 +y_6 +y_7=0$. In the latter case, since
$y_4=0$, it follows that $\by$ has support in $\{1,2,3\}$, hence
that $\bn_{123}$ is a strictly better response to $\by$ than
either $5$, $6$, or $7$; therefore, in any case, $x_5 +x_6
+x_7=0$. Since we assumed $x_1+x_2 +x_3=0$, it follows that
$\bx=\be_4$. 
Therefore, $\by$ must have support in $\{5,6,7\}$. It follows that $\bx$ is not a best-reply to $\by$, a contradiction.  

Summing up, only subcase 1.2 is possible, and then $\bx=\by=\bn_{123}$.\finpreuve
\section{Best-reply dynamics in the $7 \times 7$ game (\ref{eq:77game})} 
\label{app:BR}
This section proves Proposition  \ref{prop:universal2BR}. Recall the notation of Section \ref{sec:REP}: $\lambda$, $\mu$, $\bar{\bx}$, $\hat{\bx}$, $\bar{\bU}$ and $\hat{\bU}$. 
Consider a solution of (\ref{eq:defBR}) in game (\ref{eq:77game}) such that initially neither $x_1=x_2=x_3$ nor $x_5=x_6=x_7$. 
Thus, $\lambda(0)>0$,  $\mu(0)>0$, $\bar \bx (0)\neq (1/3, 1/3, 1/3)$ and $\hat \bx (0)\neq (1/3, 1/3, 1/3)$. This implies that $\lambda(t)$ and $\mu(t)$ are positive for all $t \geq 0$, 
as 
they can decrease at most exponentially.

We first show that, up to a change of velocity, $\bar{\bx}$ and
$\hat{\bx}$ follow the best-reply dynamics in the RPS game (\ref{eq:RPS-symbis}). Below, $BR(\cdot)$ denotes the best-reply correspondence in game (\ref{eq:RPS-symbis}).
\begin{lem} \label{lm:BR-decomp}
For almost all times $t$:
\begin{equation*} 
\dot{\bar{\bx}} \in \left(1 + \frac{\dot{\lambda}}{\lambda}\right)
\left(BR(\bar{\bx}) - \bar{\bx}\right)
\quad \mbox{ and } \quad \dot{\hat{\bx}} \in \left(1 + \frac{\dot{\mu}}{\mu}\right)\left(
BR(\hat{\bx}) - \hat{\bx}\right)
\end{equation*}
\end{lem}
\begin{proof}
We prove the first part. The proof of the second part is the same. Let $\bb \in BR(\bx(t))$ such that $\dot{\bx}(t)=\bb-\bx(t)$. 

\noindent \emph{Case 1: if $b_i=0$ for all $i=1,2,3$.} Then $\dot{x}_i=-x_i$ for all $i=1,2,3$. This implies that $\dot{\bar{\bx}}=0$ and 
that $\dot{\lambda}=-\lambda$, so that 
the result holds trivially.

\noindent \emph{Case 2.} Otherwise, define $\bar{\bb}$ as $\bar{\bx}$. A few lines of algebra show that, independently of the payoffs: 
\begin{equation}
\dot{\bar{\bx}} = \left(1
+ \frac{\dot{\lambda}}{\lambda}\right)\left(\bar{\bb} -
\bar{\bx}\right)
\end{equation}
Moreover, since all strategies in $\{1, 2, 3\}$ earn the same payoffs against strategies in $\{4, 5, 6, 7\}$, 
a variant of Lemma \ref{lm:uniqueNash-lemgen} shows that $\bar{\bb} \in BR(\bar{\bx})$, hence the result.
\end{proof}
 Recall the definition of the Shapley triangle in (\ref{eq:STsym}). 
 \begin{lem} \label{cl:barx-toST} 
If $\lambda(t)$ (resp. $\mu(t)$) does not converge to $0$,  then the limit set of $\bar{\bx}(t)$ (resp. $\hat{\bx}(t)$) is the Shapley triangle (\ref{eq:STsym}) hence $\max_i (\bar{\bU}\bar{\bx})_i \to 0$ (resp. $\max_i (\hat{\bU}\hat{\bx})_i \to 0$).
\end{lem}
 
\begin{proof}
We only prove the first part (with $\lambda(t)$). The proof of the second part is the same. Let
$\by(\cdot)$ be the unique solution of the best-reply
dynamics in game  (\ref{eq:RPS-sym}) with initial condition
$\by(0)=\bar{\bx}(0)$. Let $\tau(t)$ denote the rescaled time:
\begin{equation} \label{eq:deftau}
\tau(t):= \int_{0}^t \left(1 +
\frac{\dot{\lambda}}{\lambda}(s) \right) \, ds=t +
\ln\left(\frac{\lambda(t)}{\lambda(0)}\right)
\end{equation}
Note that $\tau(t)$ is nondecreasing as, due to
(\ref{eq:defBR}), $\dot{\lambda} \geq - \lambda$. 
Moreover $\limsup \lambda(t) >0$, hence $\tau(t) \to +\infty$ by (\ref{eq:deftau}). 
Furthermore, it follows from (\ref{lm:BR-decomp}) that for all $t \geq 0$, 
$\bar{\bx}(t)=\by(\tau(t))$. 
The result now follows from Proposition \ref{prop:BR-RPS}. 
\end{proof}

\begin{lem}
\label{lm:BR-proof77-exists} There exists a time $T>0$
such that none of the strategies $1$, $2$ and $3$ is a
best-reply to $\bx(T)$. 
\end{lem}
\begin{proof}
Assume by contradiction that for all $t \geq 0$, 
\begin{equation}
\label{eq:123good}
(i) \, (\bU\bx)_4 - \max_{1 \leq i \leq 3} (\bU \bx)_i  \leq 0 \mbox{ and } (ii) \, \max_{5 \leq i \leq 7} (\bU \bx)_i - \max_{1 \leq i \leq 3} (\bU \bx)_i \leq 0
\end{equation}
Note that the payoff of a strategy $i$ in $\{1, 2, 3\}$ may be written as 
\begin{equation}
\label{eq:123} (\bU\bx)_i=\lambda (\bar{\bU}\bar{\bx})_i - 10x_4 + \mu (-1/3 + \eps).
\end{equation} 
Similarly, for all $j$ in $\{5, 6, 7\}$,
\begin{equation}
\label{eq:567} (\bU\bx)_j= -\lambda/3 + 10x_4 + \mu (\hat{\bU}\hat{\bx})_{j'}, \mbox{with $j'=j+4$}. 
\end{equation} 
Note also that since $\be_4$ is not a best-reply to itself, $x_4(t)$ cannot converge to $1$. Now examine the following cases.
 
\emph{Case 1: if $\lambda(t) \to 0$}. Then $\mu(t) $ does not converge to $ 0$, thus it follows from $\lambda(t) \to 0$, (\ref{eq:567}) and Lemma \ref{cl:barx-toST} that
$$\limsup \max_{5 \leq i \leq 7} (\bU \bx)_i = \limsup \left(10 x_4 + \mu \max_{i} (\hat{\bU} \hat{\bx})_i\right) \geq 0$$ 
while $\limsup \max_{1 \leq i \leq 3} (\bU \bx)_i \leq - 1/3 +\eps<0$. This contradicts (ii) in (\ref{eq:123good}).

\emph{Case 2: if $\lambda(t) $ does not converge to $ 0$.} Then by Lemma \ref{cl:barx-toST}, $\bar \bx$ converges to the Shapley triangle (\ref{eq:STsym}) and there is a increasing sequence $(t_n)$ with $t_n \to +\infty$ such that $\bar \bx(t_n) \to \bar{\bq}$, where $\bar{\bq}$ is the vertex of the Shapley triangle defined in (\ref{eq:qbis}). 

\emph{Subcase 2.1. If $\mu(t) \to 0$.} Together with $\bar{\bx}(t_n) \to \bar{\bq}$, this implies that for $n$ large enough, strategy $4$ is a strictly better reply to $\bx(t_n)$ than strategies 1,2 and 3; this contradicts  part (i) of (\ref{eq:123good}).

\emph{Subcase 2.2. If $\mu(t) $ does not converge to $ 0$.} Then by Lemma \ref{cl:barx-toST} and (\ref{eq:STsym}), $\max_i (\bar{\bU}\bar{\bx})_i \to 0$ and $\max_i (\hat{\bU}\hat{\bx})_i \to 0$.  Together with (\ref{eq:123}), (\ref{eq:567}),  their analog for strategy $4$ and (\ref{eq:123good}),  this implies that along the sequence $(t_n)$:
\begin{equation}
\label{eq:a}
(\bU\bx)_4 - \max_{1 \leq i \leq 3} (\bU\bx)_i  = \lambda \left[ (\bU\bq)_4 - o(1)\right] + 10 x_4 - \eps \mu \leq 0
\end{equation}
where $\bq=(\bar{q}_1, \bar{q}_2, \bar{q}_3, 0, 0, 0, 0)$, and  
\begin{equation}
\label{eq:b}
\max_{5 \leq i \leq 7} (\bU\bx)_i - \max_{1 \leq i \leq 3} (\bU\bx)_i  = -\frac{\lambda}{3} + 20 x_4 + \mu\left(\frac{1}{3} - \eps \right)  + o(1) \leq 0
\end{equation}

Roughly, (\ref{eq:a}) implies that $\lambda/\mu$ should be small and (\ref{eq:b}) that $\lambda/\mu$ should be large. Assuming conservatively that $x_4=0$, hence $\mu=1-\lambda$, these equations may be shown to be incompatible for $\eps < 2/9$.
\end{proof}

We now conclude. By Lemma \ref{lm:BR-proof77-exists}, there exists a time $T$ such that none of the strategies $1$, $2$ and $3$ is a best-reply to $\bx(T)$. Moreover, due to Lemma \ref{lm:BR-decomp},  for all $t \geq 0$, $\hat{\bx}(t) \neq (1/3, 1/3, 1/3)$, hence by a variant of Lemma \ref{lm:uniqueNash-lemgen}, strategies $5$, $6$ and $7$ cannot all be best-replies to $\bx(t)$. Thus, due to the cyclic symmetry of strategies $5$ to $7$, we may assume that the set of pure best-replies to $\bx(T)$ is one of the followings: 

\emph{Case 1}: $\{5\}$ or $\{5,6\}$ ; \emph{Case 2}:  $\{4, 6\}$ or $\{4, 5, 6\}$ ; or \emph{Case 3}: $\{4\}$ 

In Case 1, the same arguments as in Proposition \ref{prop:66} show that $\bx(t) \to ST_{567}$. In Case 2, since strategy $4$ is strictly dominated by strategy $6$ on the face spanned by $\be_4$, $\be_5$ and $\be_6$, it immediately ceases to be a best-reply. This leads to Case 1. In Case 3, since $\be_4$ is not a best-reply to itself, there exists a first time $T'>T$ at which $\be_4$ is not the unique best-reply to $\bx(t)$. Due to the improvement principle (Lemma \ref{lm:ip}), none of the strategies $1$, $2$ and $3$ is a best-reply to $\bx(T')$. Thus we are back to Case 2. This concludes the proof. 

\end{appendix}
\begin{small}
 
\end{small}
%

%

\end{document}